\documentclass[12pt, journal, draftclsnofoot, onecolumn]{IEEEtran}

\usepackage[top=1in, bottom=1in, left=1in, right=1in]{geometry}
\usepackage[usenames]{color}
\usepackage{url}
\usepackage{subfigure}
\usepackage{amsfonts,mathrsfs,color}
\usepackage{amssymb,amsmath,amsthm}
\usepackage{verbatim}
\usepackage{acronym}
\usepackage{graphicx}
\usepackage{enumerate}
\usepackage{mathtools}
\usepackage{bbold}
\usepackage{bbm}
\usepackage{dsfont}
\RequirePackage[singlelinecheck=off]{caption}
\newcommand{\remove}[1]{}

\def\fskip#1{}

\newtheorem{theorem}{Theorem}

\newtheorem{corollary}{Corollary}

\newtheorem{definition}{Definition}
\newtheorem{example}{Example}

\newtheorem{lemma}{Lemma}

\newtheorem{remark}{Remark}

\renewcommand{\theenumi}{\arabic{enumi}}

\def\1{{\bf 1}}

\def\R{\mathbb{R}}

\begin{document}
\title{An Approximation Algorithm and Price of Anarchy for the Binary-Preference Capacitated Selfish Replication Game*}
\author{\authorblockN{Seyed Rasoul Etesami, Tamer Ba\c{s}ar}\\
  \authorblockA{Coordinated Science Laboratory, University of Illinois at Urbana-Champaign,  Urbana, IL 61801\\
     Email: (etesami1, basar1)@illinois.edu}
\thanks{*Research supported in part by the ``Cognitive \& Algorithmic Decision Making" project grant through the College of Engineering of the University of Illinois, and in part 
by AFOSR MURI Grant FA 9550-10-1-0573 and NSF grant CCF 11-11342.}
}
\maketitle
\begin{abstract}
We consider in this paper a simple model for human interactions as service providers of different resources over social networks, and study the dynamics of selfish behavior of such social entities using a game-theoretic model known as binary-preference capacitated selfish replication (CSR) game. It is known that such games have an associated ordinal potential function, and hence always admit a pure-strategy
Nash equilibrium (NE). We study the price of anarchy of such
games, and show that it is bounded above by 3; we further provide some
instances for which the price of anarchy is at least 2. We also devise a quasi-polynomial
algorithm $\mathcal{O}\big(n^{2+\ln D}\big)$ which can find, in a
distributed manner, an allocation profile that is within a constant factor of
the optimal allocation, and hence of any pure-strategy Nash equilibrium of
the game, where the parameters $n$, and $D$ denote, respectively, the
number of players, and the diameter of the network. We further show that when the underlying network has a tree structure, every globally optimal allocation is a Nash equilibrium, which can be reached in only linear time.
\end{abstract}

\section{Introduction}
In recent years, there has been a wide range of studies on the role of social and distributed networks in various disciplinary areas such as economics, computer science, epidemiology, and engineering. In particular, availability of large data from online social networks and advances in control of distributed systems have drawn the attention of many researchers to model various phenomena in social and distributed networks using some mathematical tools such as game theory in order to exploit some of the hidden properties of such networks. Such studies not only improve our understanding of the complex nature of social and distributed events, but also enable us to devise more efficient algorithms toward some desired outcomes in such networks. 

Due to accelerated growth of economic networks and advances in using game-theoretic tools in various applications, modeling of distributed network storage has become an important issue. In general, distributed network storage games or resource allocation games are characterized by a set of agents who compete for the same set of resources \cite{pacifici2012convergence,masucci2014strategic}, and arise in a wide variety of contexts such as congestion games \cite{milchtaich1996congestion,ackermann2008impact,fabrikant2004complexity}, load balancing \cite{ghosh1994dynamic}, peer-to-peer systems \cite{pollatos2008social}, web-caches \cite{gopalakrishnan2012cache}, content management \cite{pollatos2008social}, and market sharing games \cite{goemans2006market}. Among many problems that arise in such a context, one that stands out is distributed replication, which not only improves the availability of resources for users, but also increases the reliability of the entire network with respect to customer requests \cite{chun2004selfish}, \cite{goyal2000learning}. However, one of the main challenges in modeling resource allocation problems using game-theoretic tools is to answer the question of to what extent such models can predict the desired optimal allocations over a given network. Modeling a system as a game, ideally one would like for the set of equilibria or person-by-person optimal solutions of the game to be as close as possible to the desired states of the system. In fact, the \textit{price of anarchy (PoA)} \cite{koutsoupias1999worst} is one of the metrics in game theory that measures efficiency and the extent to which a system degrades due to selfish behavior of its agents; it has been used extensively in the literature \cite{pollatos2008social,goemans2006market,vetta2002nash,chun2004selfish}.          

Distributed replication games with servers that have access to all the resources and are accessible at some cost by users have been studied in \cite{laoutaris2006distributed}. Moreover, the uncapacitated selfish replication game where the agents have access to the set of all resources was studied in \cite{chun2004selfish}, where the authors were able to characterize the set of equilibrium points based on the parameters of the problem. However, unlike the uncapacitated case, there is no comparable characterization of equilibrium points in capacitated selfish replication games. In fact, when the agents have limited capacity, the situation could be much more complicated as the constraint couples the actions of agents much more than in the uncapacitated case or in replication games with servers. 

Typically, capacitated selfish replication games are defined in terms of a set of available resources for each player, where the players are allowed to communicate through an undirected communication graph. Such a communication graph identifies the access cost among the players, and the goal for each player is to satisfy his/her customers' needs with minimum cost. Ideally, and in order to avoid any additional cost, each player only wants to use his/her own set of resources. However, due to limitation on capacity, players do not have access to all the resources and hence, they incur some cost by traveling over the network and borrowing some of the resources which are not available in their own caches from others in order to meet their customers' demands. The problem of finding an equilibrium for capacitated selfish replication games in the case of hierarchical networks was studied in \cite{gopalakrishnan2012cache}. Moreover, the class of capacitated selfish replication games with binary preferences has been studied in \cite{gopalakrishnan2012cache,etesami2014pure}, where ``binary preferences" captures the behavioral pattern where players are
equally interested in some objects. 

In this paper, we consider the capacitated selfish replication game with binary preferences. In this model players act myopically and selfishly, while they are required to fully satisfy their customers' needs. Note that, although the players act in a selfish manner with respect to others in satisfying their customer needs, their actions are closely coupled with the others' and they do not have absolute freedom in the selection of their actions. It was shown in \cite{gopalakrishnan2012cache} that when the number of resources is 2, there exists a polynomial time algorithm $\mathcal{O}(n^3)$ to find an equilibrium. This result has been improved in \cite{etesami2014pure} to a linear time algorithm when the number of resources is bounded above by 5. However, in general there exist only exponential time algorithms for finding a pure-strategy Nash equilibrium. In this work we consider such games over general undirected networks and devise a quasi-polynomial algorithm which drives the system to an allocation profile whose total cost lies within a constant factor of that in any pure-strategy Nash equilibrium. In particular, we show that the price of anarchy, i.e., the ratio of the highest cost among Nash equilibrium points to the overall minimum cost of an optimal allocation, is in general bounded above by 3.    
 
The paper is organized as follows. In Section~\ref{sec:game-model}, we introduce capacitated selfish replication games with binary preferences over general undirected networks. We review some salient properties of such games and include some relevant existing results on this problem. In Section \ref{sec:Price-of-Anarchy}, we provide an upper bound on the price of anarchy of such games and also an upper bound on the optimal allocation cost. In Section \ref{sec:apx-alg} we devise a quasi-polynomial algorithm which can reach an allocation profile within a constant factor of the optimal allocation profile, and hence of any pure-strategy
Nash equilibrium. We conclude the paper with identifying future directions of research in Section~\ref{sec:conclusion}. 

\textbf{Notations}: 
For a positive integer $n$, we let $[n]:=\{1,2,\ldots,n\}$. For a vector $v\in \R^n$, we let $v_i$ be the $i$th entry of $v$. We use $\mathcal{G}=([n], \mathcal{E})$ for an undirected underlying network with a node set $\{1,2,\ldots,n\}$ and an edge set $\mathcal{E}$. For any two nodes $i, j \in [n]$, we let $d_{\mathcal{G}}(i,j)$ be the graphical distance between them, that is, the length of a shortest path which connects $i$ and $j$. The diameter of a graph, denoted by $D$, is the maximum distance between any pair of vertices, that is, $D=\max_{i,j \in [n]} d_{\mathcal{G}}(i,j)$. We let $d_{\min}$ and $d_{\max}$ denote the minimum and the maximum degree in the graph $\mathcal{G}$, respectively. Moreover, for an arbitrary node $i\in [n]$ and an integer $r\ge 0$, we define a ball of radius $r$ and center $i$ to be the set of all the nodes in the graph $\mathcal{G}$ whose graphical distance to the node $i$ is at most $r$, i.e., $B(i,r)=\{x\in \mathcal{V}| d_{\mathcal{G}}(i,x)\leq r\}$. We denote a specific Nash equilibrium and a specific optimal allocation by $P^*$ and $P^o$, respectively. Finally, we use $|S|$ to denote the cardinality of a finite set $S$.
 
\section{Problem Formulation and Existing Results}\label{sec:game-model}

In this section we first introduce the capacitated selfish replication game with binary preferences as was introduced in \cite{gopalakrishnan2012cache}. In the balance of this paper, our focus will be on such games, which for simplicity we refer to as CSR games.

\subsection{CSR Game Model}
We start with a set of $[n]=\{1,2,\ldots,n\}$ nodes (players) which are connected by an undirected graph $\mathcal{G}=([n], \mathcal{E})$. We denote the set of all resources by $O=\{o_1, o_2,\ldots, o_k\}$. For simplicity, but without much loss of generality, we assume that each node can hold only one resource in its cache. All the results can in fact be extended to CSR games with different capacities (see Remark \ref{rem:varying-cache} below). Moreover, we assume that each node has access to all the resources. For a particular allocation $P=(P_1, P_2, \ldots, P_n)$, we define the sum cost function $C_i(P)$ of the $i$th player as follows:
\begin{align}\label{eq:CSR-cost-formulation}
C_i(P)=\sum_{o\in O\setminus \{P_i\}}d_{\mathcal{G}}(i, \sigma_i(P,o)), 
\end{align}
where $\sigma_i(P,o)$ is $i$'s nearest node holding $o$ in $P$. Given an allocation profile $P$ we define the {\bf\textit{radius}} of agent $i$, denoted by $r_i(P)$, to be the distance between node $i$ and the nearest node other than her holding the same resource as $i$, i.e., $r_i(P)=\min_{ j\neq i, P_j=P_i}d_{\mathcal{G}}(i,j)$. Note that if there does not exist such a node, we simply define $r_i(P)=D$, where $D$ is the diameter of the network. We suppress the dependence of $r_i(P)$ on $P$ whenever there is no ambiguity. Finally, If some resource $o$ is missing in an allocation profile $P$, we define the cost of each player for that specific resource to be very large such as $D+1$.

\begin{remark}\label{rem:varying-cache}
Actually all the proofs in this paper can be carried over to games with varying capacities by constructing a new network which transfers games with different cache sizes to one with unit size caches \cite{etesami2014pure,gopalakrishnan2012cache}.
\end{remark}

\begin{remark}\label{rem:cost-to-radius}
Given two allocation profiles $P$ and $\tilde{P}$ which only differ in the $i$th coordinate, using \eqref{eq:CSR-cost-formulation} and the definition of the radius, one can easily see that $C_i(P)-C_i(\tilde{P})=r_i(\tilde{P})-r_i(P)$. This establishes an equivalence between decrease in cost and increase in radius for player $i$, when the actions of the other players are fixed. 
\end{remark}

It has been shown in \cite{gopalakrishnan2012cache} that the CSR game has an associated ordinal potential function, and hence, it has at least one pure Nash equilibrium. However, the number of best responses to reach an equilibrium can in general be exponentially large ($\mathcal{O}(n^D)$). Therefore, the main challenge here is to find an efficient way to arrive at an equilibrium or at least to be close enough to an equilibrium.

\subsection{Least Best Response Algorithm}\label{sec:LBR-dynamics}

It was shown earlier in \cite{etesami2014pure} that the following algorithm known as the \textit{least best response algorithm} will find a pure Nash equilibrium of the CSR game when the number of resources is low or the underlying network is ``dense enough" with respect to the number of resources.  

{\bf \textit{Least Best Response Algorithm 1:}} Given a CSR game, at each time $t=1,2,\ldots$, and from all the agents who want to update, we select an agent with the least radius and let her update be based on her best response. Ties are broken arbitrarily.  

\begin{theorem}\label{thm:dense-graph}(\cite{etesami2014pure})
In the binary CSR game with $n$ agents and network diameter $D$, let $T$ denote the convergence time of the least best response algorithm to a pure Nash equilibrium. Then,
\begin{align}\nonumber
T\leq \begin{cases} 3n^3\min\{|O|-1,D\}, & \mbox{if } \ \ |O|\leq d_{\min} \\ 
n\min\{D, |O|-1\}, & \mbox{if} \ \ \ |O|< 5, \end{cases}
\end{align} 
\end{theorem}

In the remainder of this paper, we provide substantial improvements to these results. We provide a constant approximation algorithm which works in quasi-polynomial time over general networks and without any extra assumption on the structure of the network. Roughly speaking, this algorithm reduces the time complexity of being close to an equilibrium from the naive search of $\mathcal{O}(n^{D})$ to $\mathcal{O}(n^{\ln D})$.

\section{Price of Anarchy for the CSR Game}\label{sec:Price-of-Anarchy}

In this section, we show that all the equilibrium points of the CSR game are almost as good as the global optimal allocation, i.e., an allocation profile which minimizes the total sum of the costs of the players.

\begin{theorem}\label{thm:Price-of-Anarchy}
The price of anarchy in the CSR game is bounded from above by 3. 
\end{theorem}
\begin{proof}
For an arbitrary equilibrium $P^*$ and a specific node $i$ with equilibrium radius $r^*_i$, i.e., $r^*_i=d_{\mathcal{G}}(i,\sigma_i(P^*,P^*_i))$, we note that all the resources must appear at least once in $B(i,r^*_i)$. In fact, if a specific resource is missing in $B(i,r^*_i)$, then node $i$ can increase its radius by updating its current resource to that specific resource, thereby decreasing its cost. But this is in contradiction with $P^*$ being an equilibrium (Figure \ref{fig:them-POA}). Now, given the equilibrium profile $P^*$, let us define $\hat{r}_{i}$ to be the smallest integer such that $B(i,\hat{r}_{i})$ contains at least two resources of the same type, i.e.,

\begin{align}\nonumber
\hat{r}_{i}=\min\Big\{r\in \mathbb{N}: &\ \forall j,k\in B(i,r\!-\!1), P^*_j\neq P^*_k, \mbox{and} \ \exists j_0,k_0\in B(i,r), P^*_{j_0}=P^*_{k_0} \Big\}.
\end{align} 

\begin{figure}[htb]
\vspace{-1cm}
\begin{center}
\includegraphics[totalheight=.25\textheight,
width=.35\textwidth,viewport=-100 0 1100 1200]{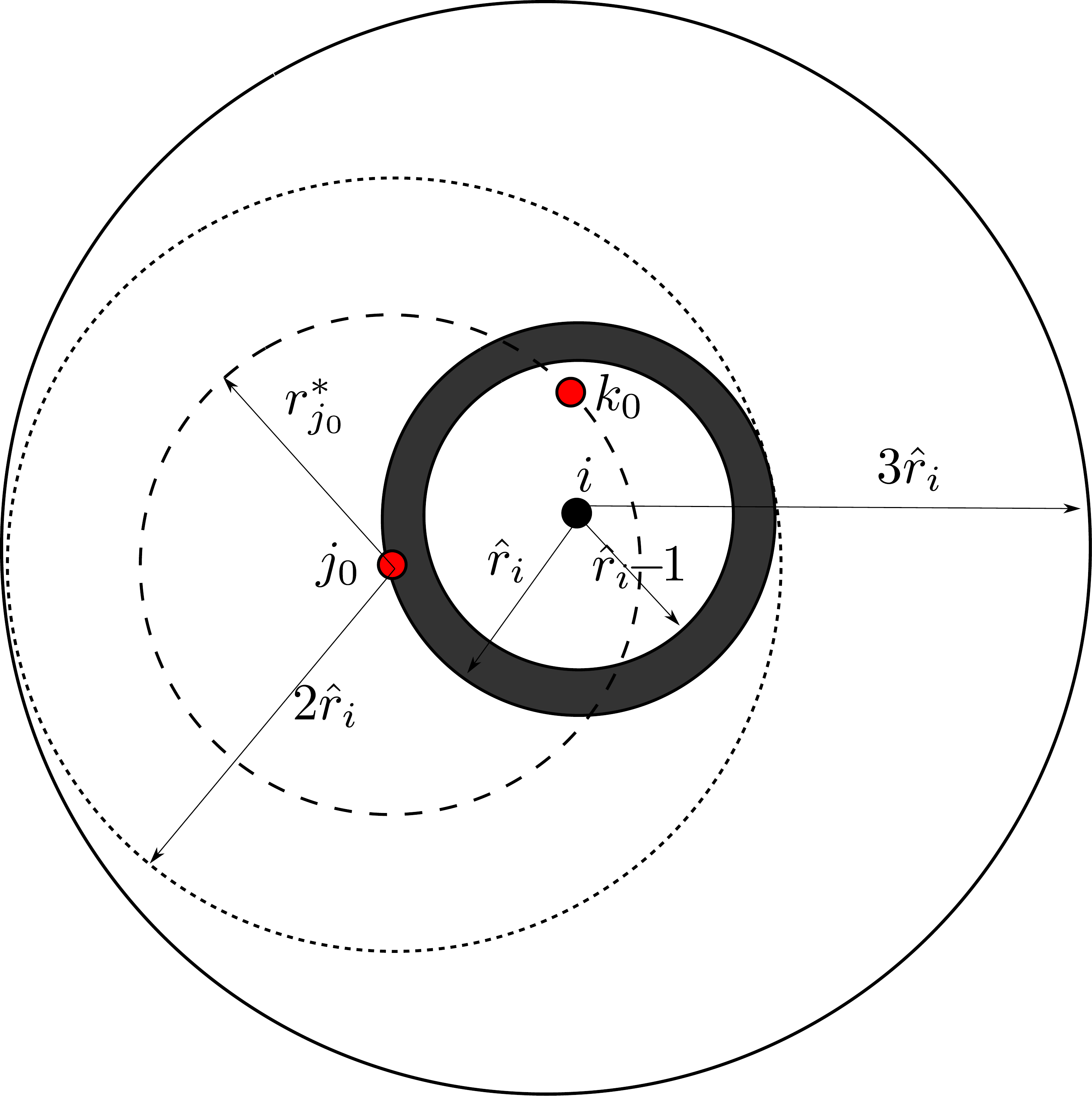} 
\end{center}
\vspace{-0.4cm}\caption{Illustration of resource allocation in Nash equilibrium $P^*$. Note that $r^*_{j_0}$ denotes the radius of node $j_0$ at equilibrium $P^*$ which is upper bounded by $d_{\mathcal{G}}(j_0,k_0)$. The ball $B(i,\hat{r}_{i}-1)$ contains only resources of different types.}
\label{fig:them-POA}
\end{figure}

Now we claim that all the resources must appear at least once in $B(i,3\hat{r}_{i})$. To see this and by the above definition, let $j_0\neq k_0\in B(i,\hat{r}_{i})$ be such that $P^*_{j_0}=P^*_{k_0}$. This means that the equilibrium radius of node $j_0$, i.e., $r^*_{j_0}$ in $P^*$ is at most $d_{\mathcal{G}}(j_0,i)+d_{\mathcal{G}}(i,k_0)\leq 2\hat{r}_{i}$. On the other hand, by the argument at the beginning of the proof, all the resources must appear at least once in $B(j_0,2\hat{r}_{i})$. But since $B(j_0,2\hat{r}_{i})\subseteq B(i,3\hat{r}_{i})$, this shows that $B(i,3\hat{r}_{i})$ must include all the resources at least once. 

Next, let us denote an optimal allocation profile by $P^o$, and the cost of node $i$ in the optimal allocation and Nash equilibrium by $C_i(P^o)$ and $C_i(P^*)$, respectively. Now for the equilibrium $P^*$ and since by the definition of $\hat{r}_{i}$ there are no two similar resources in $B(i,\hat{r}_{i}\!-\!1)$, and all the resources appear at least once in $B(i,3\hat{r}_{i})$, we can write 
\begin{align}\label{eq:nash-anarchy}
C_i(P^*)\leq \!\!\!\!\!\!\sum_{j\in B(i,\hat{r}_{i}\!-\!1)}\!\!\!\!\!\!\!\!d_{\mathcal{G}}(i,j)\!+\!3\hat{r}_{i}\!\left(|O|\!-\!1\!-\!|B(i,\hat{r}_{i}\!-\!1)|\right)\!.
\end{align}
On the other hand, for the cost of node $i$ in the optimal allocation $P^o$ we can write,
\begin{align}\label{eq:optimal-anarchy}
C_i(P^o)\ge \!\!\!\!\!\!\sum_{j\in B(i,\hat{r}_{i}\!-\!1)}\!\!\!\!\!\!\!\!d_{\mathcal{G}}(i,j)\!+\!\hat{r}_{i}\!\left(|O|\!-\!1\!-\!|B(i,\hat{r}_{i}\!-\!1)|\right)\!,
\end{align} 
where the inequality holds since node $i$ has to pay at least $\sum_{j\in B(i,\hat{r}_{i}\!-\!1)}d_{\mathcal{G}}(i,j)$ for the first $|B(i,\hat{r}_{i}\!-\!1)|$ closest resources, and to pay at least $\hat{r}_{i}$ for the remaining $(|O|-1-|B(i,\hat{r}_{i}\!-\!1)|)$ resources. By comparing relations \eqref{eq:nash-anarchy} and \eqref{eq:optimal-anarchy}, it is not hard to see that $\frac{C_i(P^*)}{C_i(P^o)}$ is at most
\begin{align}\nonumber
\frac{\sum_{j\in B(i,\hat{r}_{i}\!-\!1)}d_{\mathcal{G}}(i,j)+3\hat{r}_{i}\left(|O|-1-|B(i,\hat{r}_{i}\!-\!1)|\right)}{\sum_{j\in B(i,\hat{r}_{i}\!-\!1)}d_{\mathcal{G}}(i,j)+\hat{r}_{i}\left(|O|-1-|B(i,\hat{r}_{i}\!-\!1)|\right)},  
\end{align} 
which is bounded from above by 3. Since node $i$ and the equilibrium $P^*$ were chosen arbitrarily, for every equilibrium $P^*$ and for all $i\in V(\mathcal{G})$, we have $C_i(P^*)\leq 3C_i(P^o)$. Summing this inequality over all $i\in V(\mathcal{G})$ we get $ C(P^*)=\sum_i C_i(P^*)\leq 3\sum_i C_i(P^o)=3C(P^o)$. Since we have this inequality for all possible Nash equilibria, and using the definition of the price of anarchy, we have $PoA=\frac{\max_{P^*\in NE} C(P^*)}{C(P^o)}\leq 3$. 
\end{proof}

Next in the following we provide an example which shows that for some network configurations, the price of anarchy of the CSR game can actually be arbitrarily close to 2.
\begin{example}\label{ex:POA}
Given arbitrary positive integers $m,|O|\ge 2$, let us consider a network of $n:=(m+1)(|O|-1)$ nodes (players) as shown in Figure \ref{fig:example-POA}. Here we assume that $|O|$ is the number of resources $\{o_1,o_1,\ldots,o_{|O|}\}$ in the CSR game over this network. The bottom part of the network is composed of a clique of $|O|-1$ nodes, and the top part forms an independent set of $m(|O|-1)$ nodes which all are connected to the bottom part. As it can be seen in Figure \ref{fig:example-POA}, the top figure constitutes a pure Nash equilibrium for the CSR game with the total cost of $C(P^*)=m(|O|-1)(2|O|-3)+(|O|-1)^2$. Moreover, it can easily be seen that the bottom figure illustrates an optimal allocation, where the cost of each node is $|O|-1$, and hence the total optimal cost equals $C(P^{o})=(m+1)(|O|-1)^2$. Thus, we have $\frac{C(P^*)}{C(P^{o})}=2-\left(\frac{m+|O|-1}{(m+1)(|O|-1)}\right)$. This shows that for large $m$ and $|O|$, the price of anarchy of the CSR game over such networks can be arbitrarily close to 2.  
\end{example}

\begin{figure}[htb]
\vspace{2.3cm}
\begin{center}
\includegraphics[totalheight=.15\textheight,
width=.25\textwidth,viewport=100 0 700 600]{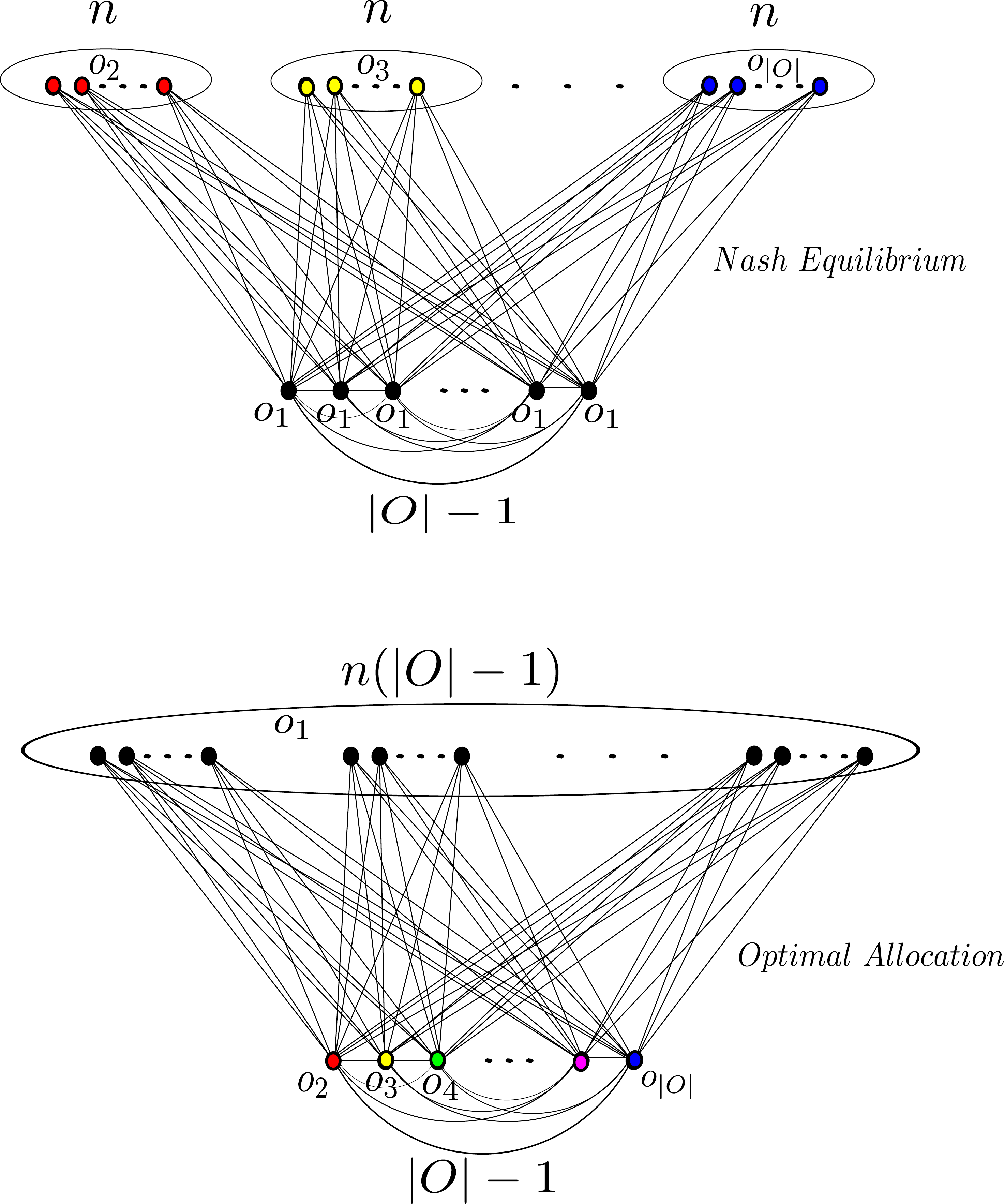} \hspace{0.4in}
\end{center}
\vspace{-0.25cm}\caption{Illustration of the resource allocation for a Nash equilibrium (top figure), and an optimal allocation (bottom figure) in Example \ref{ex:POA}.}
\label{fig:example-POA}
\end{figure}

Before concluding, we provide in the theorem below an upper bound on the minimum social cost in the general CSR game, which uses a probabilistic argument.  
 
\begin{theorem}\label{thm:upper-bound-randomization-optimal}
For a network $\mathcal{G}$ with vertex degrees $d_i, i\in [n]$, the optimal allocation cost for the CSR game is bounded above by $n(2|O|-1)+|O|(|O|-1)\sum_{i=1}^{n}\left(1-\frac{1}{|O|}\right)^{d_i}$. 
\end{theorem}
\begin{proof}
Let us assign with probability $\frac{1}{|O|}$ and independently a resource to any of the vertices of $\mathcal{G}$. Now for an arbitrary but fixed node $i$, $d_{\mathcal{G}}(i,\sigma_i(o,P))$ is a random variable denoting the graphical distance between player $i$ and the player closest to him who has resource $o$ under such a random resource allocation $P$. For any nonnegative integer $r$ and any arbitrary but fixed resource $o\in O$, we have

\begin{align}\nonumber
\mathbb{P}\big\{d_{\mathcal{G}}(i,\sigma_i(o,P))\ge r\big\} = \begin{cases}\!\left(1\!-\!\frac{1}{|O|}\right)^{|B(i,r-1)|} & \!\!\!\mbox{if} \ 0\leq r\leq D\!+\!1 \\ 
0 & \!\!\!\mbox{else}. \end{cases}
\end{align}

Note that for $r=0$ we have $B(i,r-1)=\emptyset$, and hence, $|B(i,r-1)|=0$. 
Therefore, using the tail formula for expectation of discrete random variables, the expected distance of player $i$ from resource $o$ is $\sum_{r=1}^{\infty}\mathbb{P}(d_{\mathcal{G}}(i,\sigma_i(o,P))\ge r)$. Hence, the total expected cost of such random assignment for all players and all resources  is

\begin{align}\nonumber
&\mathbb{E}[C(P)]=\sum_{o\in O}\sum_{i=1}^{n}\sum_{r=0}^{\infty}\mathbb{P}(d_{\mathcal{G}}(i,\sigma_i(o,P))\ge r)\cr
&\qquad=|O|\sum_{i=1}^{n}\sum_{r=0}^{D+1}\left(1-\frac{1}{|O|}\right)^{|B(i,r-1)|}\cr 
&\qquad=|O|\sum_{i=1}^{n} \left(1+\left(1-\frac{1}{|O|}\right)+\sum_{r=2}^{D+1} \left(1-\frac{1}{|O|}\right)^{|B(i,r-1)|}\right)\cr
&\qquad<n(2|O|-1)+|O|\sum_{i=1}^{n} \sum_{k=d_i+1}^{\infty}\left(1-\frac{1}{|O|}\right)^{k}\cr &\qquad=n(2|O|-1)+|O|(|O|-1)\sum_{i=1}^{n}\left(1-\frac{1}{|O|}\right)^{d_i}.
\end{align} 

This shows that there exists at least one resource assignment where the total cost is at most $n(2|O|-1)+|O|(|O|-1)\sum_{i=1}^{n}\left(1-\frac{1}{|O|}\right)^{d_i}$.    
\end{proof}
\begin{remark}
One can easily check that for the line graph, the bound in Theorem \ref{thm:upper-bound-randomization-optimal} is tight up to a constant factor.
\end{remark}

\section{Quasi-Polynomial Approximation Algorithm}\label{sec:apx-alg}

In this section we show that for the CSR game over general networks, one can obtain an allocation profile whose total cost lies within a constant factor of that in an optimal allocation in only quasi-polynomial time. We start with the following lemma. 
      
\begin{lemma}\label{lemm:Alg-lemma}
Given a number $\epsilon>1$, at every time instance there exists an agent $i$ who can increase its radius by a factor of at least $\epsilon$ by playing its best response. Dynamics thus generated terminate after no longer than $\mathcal{O}\big(n^2D^{\log_{\epsilon}n}\big)$ steps.  
\end{lemma}
\begin{proof}
Given a profile $P(t)=(P_1(t),\ldots, P_n(t))$ at time $t$, let us denote the radii of the players by $r_1(t), \ldots, r_n(t)$, and let $n_k(t)$ be the number of players whose radii are $k$, for $k=1,2,\ldots,D$. Now, given that at time step $t$ player $i$ with $P_i(t)=o_1$ wants to change to $P_i(t+1)=o_2$, it means that at time step $t$, there is no resource of type $o_2$ in $B(i,\epsilon r_i)$ (Figure \ref{fig:lemma-alg}). Now, let us define a potential function $R(t)$ to be
\begin{align}\label{lemm:potential}
R(t):=\sum_{k=1}^{D}\frac{n_k(t)}{k^{\log_{\epsilon} n}}=\sum_{k=1}^{n}\frac{1}{(r_k(t))^{\log_{\epsilon} n}},
\end{align}
where the equality holds since by definition of $n_k(t)$, exactly $n_k(t)$ of the terms in $\sum_{k=1}^{n}\frac{1}{(r_k(t))^{\log_{\epsilon} n}}$ are equal to $\frac{1}{k^{\log_{\epsilon} n}}$. Moreover, one can easily see that $R(\cdot)$ is a nonnegative function which is upper bounded by $n$. We will show that after each time of running the dynamics, the value of the potential function given in \eqref{lemm:potential} decreases by at least $\frac{1}{n\|\{r_{\max}(t)\}\|_{\infty}^{\log_{\epsilon}n}}$, where $\|\{r_{\max}(t)\}\|_{\infty}=\max_{t\ge 0}\max_{i\in [n]}r_i(t)$. To see this, let us assume that node $i$ updates its radius from $r_i(t)$ to $r_i(t+1)\ge \epsilon r_i(t)$. Therefore, for some $k\in [n]$ we must have one of the following three cases:
\begin{itemize}
\item If $k=i$, then $r_i(t+1)\ge \epsilon r_i(t)$. This holds due to the dynamics rule, since node $i$ updates its radius from $r_i(t)$ to $r_i(t+1)$ if and only if $r_i(t+1)\ge \epsilon r_i(t)$.

\item If $r_k(t)< r_i(t+1)$, then $r_k(t+1)\ge r_k(t)$. To see this, note that if $r_k(t)<r_i(t+1)$, then $P_k(t)\neq o_2$. Moreover, either $P_k(t)= o_1$, where in this case by updating $P_i(t)$ from $P_i(t)=o_1$ to $P_i(t+1)=o_2$, the radius of node $k$ does not decrease, i.e., $r_k(t+1)\ge r_k(t)$, or $P_k(t)\neq o_1$, where in this case the radius of node $k$ remains the same, i.e., $r_k(t+1)=r_k(t)$.   

\item If $r_k(t)\ge r_i(t+1)$, then $r_i(t+1)\leq r_k(t+1)$. To show this, first note that if $P_k\neq o_1,o_2$, then the radius of node $k$ does not change after node $i$ updates, i.e., $r_k(t+1)=r_k(t)\ge r_i(t+1)$. Otherwise, either $P_k(t)=o_1$, which in this case $r_k(t+1)\ge r_k(t)\ge r_i(t+1)$ due to the fact that after update we have fewer number of $o_1$-resources, or $p_k(t)=o_2$ in which case the radius of player $k$ cannot decrease to less than the graphical distance between $k$ and $i$, thus $r_k(t\!+\!1)\!\ge\! r_k(t)\!\ge\! r_i(t\!+\!1)$.  
\end{itemize}
Using the fact that $r_i(t+1)\ge \epsilon r_i(t)$ and considering the above three possibilities, we can write

\begin{figure}[htb]
\vspace{-2.3cm}
\begin{center}
\includegraphics[totalheight=.3\textheight,
width=.45\textwidth,viewport=-150 0 700 800]{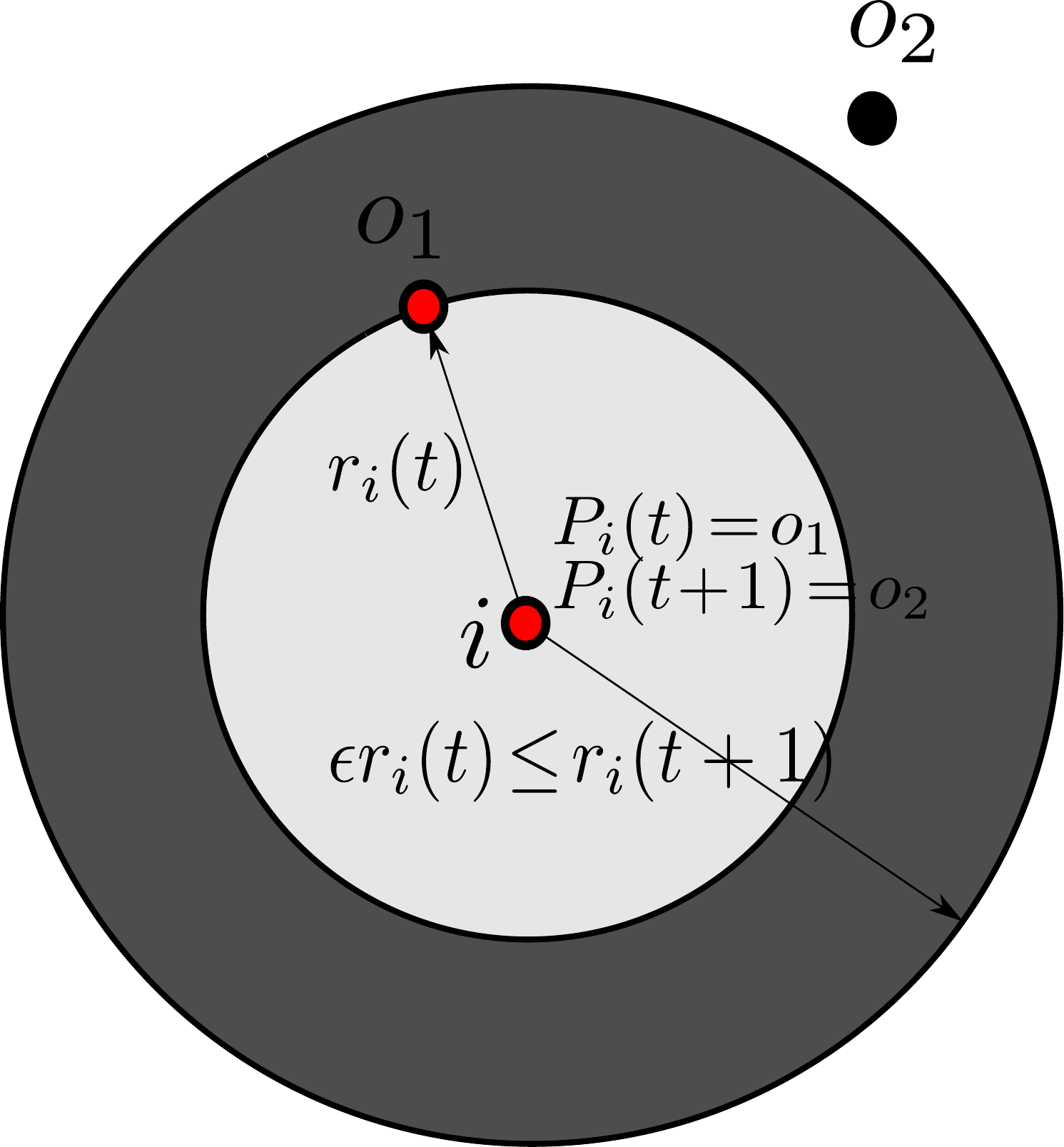} \hspace{0.4in}
\end{center}
\vspace{-0.4cm}\caption{Illustration of updating node $i$ in the Lemma \ref{lemm:Alg-lemma}. Note that there no other resource of type $o_2$ in $B(i,r_i(t\!+\!1))$.}
\label{fig:lemma-alg}
\end{figure}

\vspace{-0.7cm}
\begin{align}\nonumber
&R(t+1)-R(t)=\sum_{k=1}^{n}\left(\frac{1}{(r_k(t+1))^{\log_{\epsilon} n}}-\frac{1}{(r_k(t))^{\log_{\epsilon} n}}\right)\cr 
&=\frac{1}{(r_i(t+1))^{\log_{\epsilon} n}}-\frac{1}{(r_i(t))^{\log_{\epsilon}n}}\cr
&+\!\!\!\!\!\!\sum_{\{k\neq i: r_k(t)< r_i(t+1)\}}\left(\frac{1}{(r_k(t+1))^{\log_{\epsilon} n}}-\frac{1}{(r_k(t))^{\log_{\epsilon} n}}\right)\cr 
&+\!\!\!\!\!\sum_{\{k: r_k(t)\ge r_i(t+1)\}}\left(\frac{1}{(r_k(t+1))^{\log_{\epsilon} n}}-\frac{1}{(r_k(t))^{\log_{\epsilon} n}}\right)\cr
&\leq (\frac{1}{n}-1)\left(\frac{1}{(r_i(t))^{\log_{\epsilon}n}}\right)+0+\frac{|\{k: r_k(t)\ge r_i(t+1)\}|}{n(r_i(t))^{\log_{\epsilon}n}}\cr 
&\leq (\frac{1}{n}-1)\left(\frac{1}{(r_i(t))^{\log_{\epsilon}n}}\right) 
+\frac{n-2}{n(r_i(t))^{\log_{\epsilon}n}}\cr 
&=\frac{-1}{n(r_i(t))^{\log_{\epsilon}n}}\leq \frac{-1}{n\|\{r_{\max}(t)\}\|_{\infty}^{\log_{\epsilon}n}},  
\end{align}

\hspace{-0.34cm}where in the second to last inequality, and without any loss of generality, we can consider $|\{k: r_k(t)\ge r_i(t+1)\}|\leq n-2$. In fact, we argue that at most for $D$ instances we can have $|\{k: r_k(t)\ge r_i(t+1)\}|\ge n-2$, which does not really change the quasi-polynomial order of the termination time. Note that at least $i\notin \{k: r_k(t)\ge r_i(t+1)\}$, thus $|\{k: r_k(t)\ge r_i(t+1)\}|\leq n-1$. Moreover $|\{k: r_k(t)\ge r_i(t+1)\}|$ cannot be equal to $n-1$ more than $D$ steps, since by Lemma 1 in \cite{etesami2014pure} every time that $|\{k: r_k(t)\ge r_i(t+1)\}|= n-1$, then after updating node $i$ at the next time step the minimum index of the positive entries in $n(t+1)=(n_1(t+1),n_2(t+1),\ldots,n_D(t+1))$ will increase by at least 1, which cannot happen for more than $D$ steps.

Finally, since $R(\cdot)$ is upper bounded by $n$, $R(\cdot)$ cannot decrease more that $n^2\|\{r_{\max}(t)\}\|_{\infty}^{\log_{\epsilon}n}$ times, which shows that the dynamics must terminate after at most $\mathcal{O}\Big(n^2\|\{r_{\max}(t)\}\|_{\infty}^{\log_{\epsilon}n}\Big)$ steps. Moreover, since for all $t=0,1,\ldots$ and $i\in [n]$, $r_i(t)$ cannot exceed the diameter $D$ of the network,  $\|\{r_{\max}(t)\}\|_{\infty}\leq D$. This shows that the dynamics must terminate after at most $\mathcal{O}\Big(n^2D^{\log_{\epsilon}n}\Big)$ steps.           
\end{proof}

Next we introduce a useful definition.

\begin{definition}
Given an undirected network $\mathcal{G}=([n],\mathcal{E})$, and an allocation profile $P$, the \textit{resource-radius} of any node $i$ is the smallest positive integer $\gamma_i$ such that all the resources appear at least once in $B(i,\gamma_i)$ for that given profile $P$. 
\end{definition}

In what follows, we introduce an algorithm, called \textit{$\epsilon$-best response algorithm}, which in at most quasi-polynomial time can find an allocation profile which lies within a constant factor of an optimal allocation in the CSR game.

{\bf \textit{$\epsilon$-Best Response Algorithm 2:}} Given a network $\mathcal{G}=([n], \mathcal{E})$, a number  $\epsilon>1$, and an arbitrary initial allocation profile $P(0)$, at every time instance we select an agent $i$ who can increase its radius $r_i(t)$ by a factor of at least $\epsilon$, and let her to play her best response (ties are broken arbitrarily).         

\begin{example}\label{ex:alg}
Through this example, we illustrate how the $\epsilon$-best response algorithm works. Let $\epsilon=2$, and consider a network of 10 nodes and 4 available resources $O=\{o_1,o_2,o_3,o_4\}$. The initial profile of allocated resources has been illustrated in Figure \ref{fig:example-alg}. Let us assume that at the first time instant node $i$ has been selected by the algorithm. Since the radius of node $i$ in the initial profile is 1, i.e., $r_i(0)=1$, among the resources whose distances to node $i$ are at least $\epsilon r_i(0)=2\times 1$(in this example $\{o_3,o_4\}$), she will play her best response, i.e., $o_3$. Therefore, $P_i(1)=o_3$. Now at the second time instant, and given that node $j$ is selected, since $r_j(1)=1$, among the resources whose distances to $j$ are at least $\epsilon r_j(1)=2\times 1$ (in this example only $\{o_4\}$), she will play her best response, i.e., $P_j(2)=o_4$. Note that if the set of available resources where an agent is supposed to choose her best resource from is empty at some time instant, then she will not update.
     
\begin{figure}[htb]
\vspace{-1.75cm}
\begin{center}
\vspace{-0.75cm}
\includegraphics[totalheight=.3\textheight,
width=.45\textwidth,viewport=0 0 650 600]{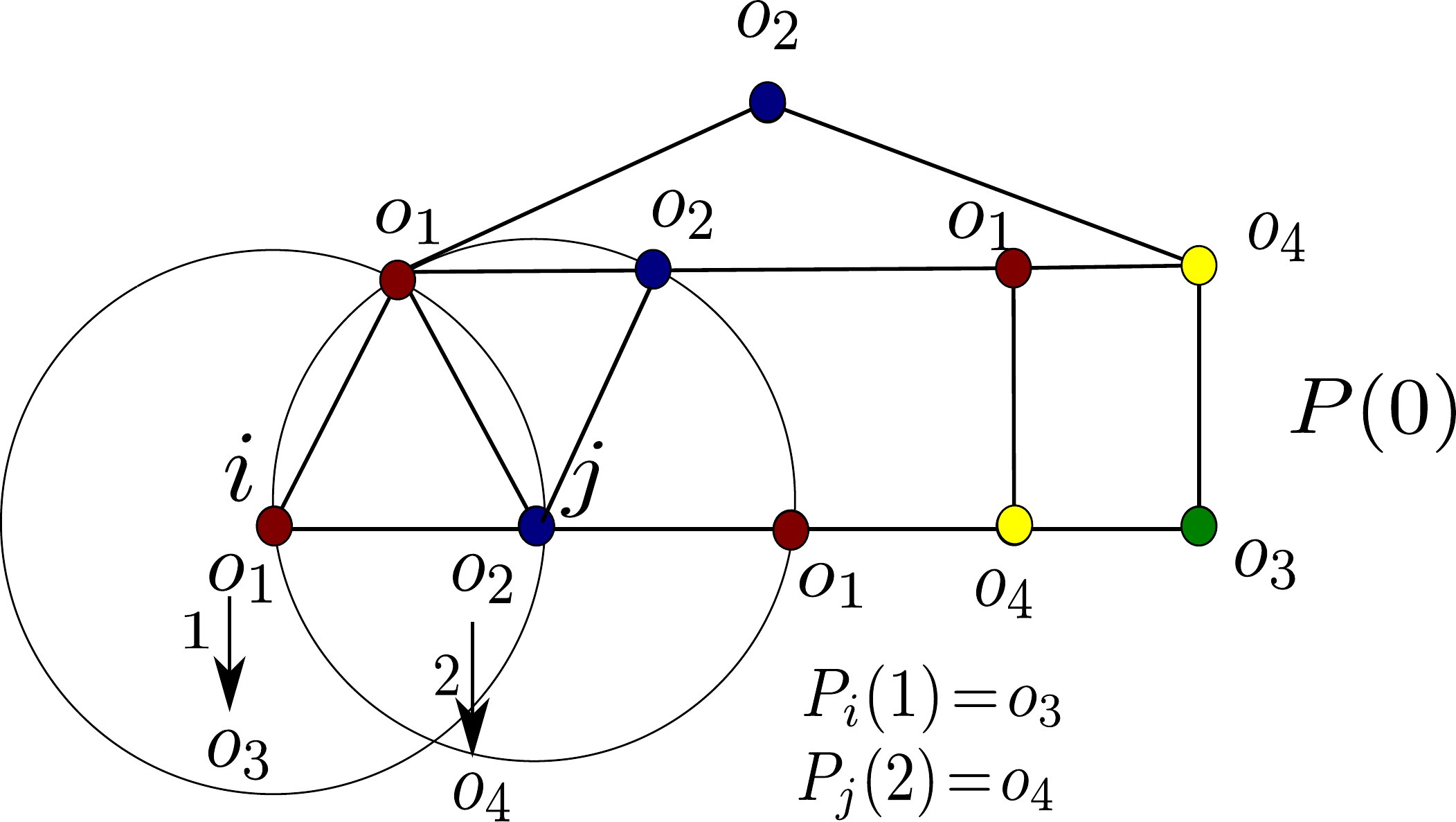} \hspace{0.4in}
\end{center}
\vspace{-0.2cm}\caption{Resource allocation in the $\epsilon$-best response algorithm (Example \ref{ex:alg}).}
\label{fig:example-alg}
\end{figure}   
\end{example}

In the following theorem we prove that the allocation profile reached after executing Algorithm 2 is a constant approximation of the optimal allocation in the CSR game.

\begin{theorem}\label{thm:approximation}
The $\epsilon$-best response algorithm provides a $(2\epsilon+1)$-approximation of the optimal allocation after at most $\mathcal{O}\Big(n^2D^{\log_{\epsilon}n}\Big)$ steps.
\end{theorem}
\begin{proof}
First, we note that by Lemma \ref{lemm:Alg-lemma} the $\epsilon$-best response algorithm terminates after at most $\mathcal{O}\Big(n^2D^{\log_{\epsilon}n}\Big)$ steps. Therefore, we only need to show that the final profile $\overline{P}$ is indeed within a constant factor of the optimal allocation. For any arbitrary $i\in [n]$, suppose that $\gamma_i$ is the resource-radius of node $i$ at the final profile $\overline{P}$. We now claim that $B(i, \frac{\gamma_i-1}{2\epsilon+1})$ contains only different resources. To see this, let us assume by contradiction that there are two nodes $j,j'\in B(i, \frac{\gamma_i-1}{2\epsilon+1})$ which share the same resource, i.e., $\overline{P}_j=\overline{P}_{j'}$. Since $j,j'\in B(i, \frac{\gamma_i-1}{2\epsilon+1})$, $d_{\mathcal{G}}(j,j')\leq \frac{2(\gamma_i-1)}{2\epsilon+1}$. This shows that $r_j\leq \frac{2(\gamma_i-1)}{2\epsilon+1}$, and hence, $\epsilon r_j\leq \frac{2\epsilon(\gamma_i-1)}{2\epsilon+1}$. In particular, since $d_{\mathcal{G}}(j,i)\leq \frac{\gamma_i-1}{2\epsilon+1}$, $B(j,\epsilon r_j)\subseteq B(i,\frac{2\epsilon(\gamma_i-1)}{2\epsilon+1}+\frac{\gamma_i-1}{2\epsilon+1})=B(i,\gamma_i-1)$. Since we assumed that $\overline{P}$ is the final allocation of the algorithm, every resource must appear at least once in $B(j,\epsilon r_j)$, as otherwise node $j$ can update its resource to another one which is outside of $B(j,\epsilon r_j)$. But, by definition of resource-radius $\gamma_i$, we know that for the profile $\overline{P}$ there exists at least one resource which is missing in $B(i,\gamma_i-1)$. This is in contradiction with $B(j,\epsilon r_j)\subseteq B(i,\gamma_i-1)$, which implies that all the resources in $B(i, \frac{\gamma_i-1}{2\epsilon+1})$ are different.
          
Finally, as we have shown above, for every agent $i$, the final profile $\overline{P}$ at the termination of Algorithm 2 has the property that all the vertices in $B(i,\frac{\gamma_i-1}{2\epsilon+1})$ have different resources, while all the resources appear at least once in $B(i,\gamma_i)$. Thus we have 
\begin{align}\nonumber
C_i(\overline{P})\leq\!\!\!\!\!\!\sum_{j\in B(i,\frac{\gamma_i-1}{2\epsilon+1})}\!\!\!\!\!d_{\mathcal{G}}(i,j)+\gamma_i\left(\!|O|\!-\!1\!-\!|B(i,\frac{\gamma_i-1}{2\epsilon+1})|\!\right)\!\!,
\end{align}
where the inequality holds due to the definition of $\gamma_i$. On the other hand, for the cost of node $i$ in the optimal placement, we can write
\begin{align}\nonumber
C_i(P^o)&\ge\!\!\!\!\!\!\!\!\sum_{j\in B(i,\frac{\gamma_i-1}{2\epsilon+1})}\!\!\!\!\!d_{\mathcal{G}}(i,j)\!+\!(\frac{\gamma_i}{2\epsilon\!+\!1})\!\left(\!|O|\!-\!1\!-\!|B(i,\frac{\gamma_i-1}{2\epsilon\!+\!1})|\!\right)\!, 
\end{align} 
where the above inequality holds since node $i$ in the optimal allocation $P^o$ has to pay at least $\sum_{j\in B(i,\frac{\gamma_i-1}{2\epsilon+1})}d_{\mathcal{G}}(i,j)$ for the first $|B(i,\frac{\gamma_i-1}{2\epsilon+1})|$ closest resources, and to pay an integer cost strictly larger than $\frac{\gamma_i-1}{2\epsilon+1}$ (at least $(\frac{\gamma_i}{2\epsilon+1})$), for the remaining $(|O|-1-|B(i,\frac{\gamma_i-1}{2\epsilon+1})|)$ resources. By comparing the above two relations, it is not hard to see that for all $i\in [n]$, we have $C_i(\overline{P})\leq (2\epsilon+1)C_i(P^o)$. Summing all of these inequalities for $i\in [n]$, we get $C(\overline{P})\leq (2\epsilon+1)C(P^o)$.                  
\end{proof} 

\begin{corollary}
Replacing $\epsilon=e$ in the result of Theorem \ref{thm:approximation}, one can easily see that Algorithm 2 gives us an allocation profile within a constant factor of the optimal allocation and any NE after at most $\mathcal{O}\left(n^2D^{\ln n}\right)=\mathcal{O}\left(n^{2+\ln D}\right)$ steps.  
\end{corollary}

Before concluding, we state the following theorem for the case where the underlying network has a tree structure. 
We omit the proof here due to space limitation, but we include it in the full version of the paper \cite{etesami2015approximation}.

\begin{theorem}\label{thm:label-optimal}
Assume that the network in the CSR game is a tree of $n$ nodes. Then every optimal allocation $P^o$ must be a NE. Furtheremore, there is an algorithm which reaches an optimal allocation in only $n$ steps.    
\end{theorem}

\section{Conclusion}\label{sec:conclusion}
In this paper, we have studied the binary-preference capacitated selfish replication (CSR) game over general networks. We have provided a distributed quasi-polynomial time algorithm $\mathcal{O}\Big(n^{2+\ln D}\Big)$ which can approximate any NE of the system within a constant factor. Moreover, we have shown that such games benefit from having a low price of anarchy.   

As an avenue for future research, one possibility is to study the dynamic version of the $CSR$ game, in the spirit of what has been discussed in \cite{etesami2014game} and  \cite{fabrikant2003network}.

\bibliographystyle{IEEEtran}
\bibliography{thesisrefs}

\begin{thebibliography}{10}
\providecommand{\url}[1]{#1}
\csname url@samestyle\endcsname
\providecommand{\newblock}{\relax}
\providecommand{\bibinfo}[2]{#2}
\providecommand{\BIBentrySTDinterwordspacing}{\spaceskip=0pt\relax}
\providecommand{\BIBentryALTinterwordstretchfactor}{4}
\providecommand{\BIBentryALTinterwordspacing}{\spaceskip=\fontdimen2\font plus
\BIBentryALTinterwordstretchfactor\fontdimen3\font minus
  \fontdimen4\font\relax}
\providecommand{\BIBforeignlanguage}[2]{{%
\expandafter\ifx\csname l@#1\endcsname\relax
\typeout{** WARNING: IEEEtran.bst: No hyphenation pattern has been}%
\typeout{** loaded for the language `#1'. Using the pattern for}%
\typeout{** the default language instead.}%
\else
\language=\csname l@#1\endcsname
\fi
#2}}
\providecommand{\BIBdecl}{\relax}
\BIBdecl

\bibitem{pacifici2012convergence}
V.~Pacifici and G.~Dan, ``Convergence in player-specific graphical resource
  allocation games,'' \emph{Selected Areas in Communications, IEEE Journal on},
  vol.~30, no.~11, pp. 2190--2199, 2012.

\bibitem{masucci2014strategic}
A.~M. Masucci and A.~Silva, ``Strategic resource allocation for competitive
  influence in social networks,'' \emph{arXiv preprint arXiv:1402.5388}, 2014.

\bibitem{milchtaich1996congestion}
I.~Milchtaich, ``Congestion games with player-specific payoff functions,''
  \emph{Games and economic behavior}, vol.~13, no.~1, pp. 111--124, 1996.

\bibitem{ackermann2008impact}
H.~Ackermann, H.~R{\"o}glin, and B.~V{\"o}cking, ``On the impact of
  combinatorial structure on congestion games,'' \emph{Journal of the ACM
  (JACM)}, vol.~55, no.~6, p.~25, 2008.

\bibitem{fabrikant2004complexity}
A.~Fabrikant, C.~Papadimitriou, and K.~Talwar, ``The complexity of pure {N}ash
  equilibria,'' in \emph{Proceedings of the thirty-sixth annual ACM symposium
  on Theory of computing}.\hskip 1em plus 0.5em minus 0.4em\relax ACM, 2004,
  pp. 604--612.

\bibitem{ghosh1994dynamic}
B.~Ghosh and S.~Muthukrishnan, ``Dynamic load balancing in parallel and
  distributed networks by random matchings,'' in \emph{Proceedings of the sixth
  annual ACM symposium on Parallel algorithms and architectures}.\hskip 1em
  plus 0.5em minus 0.4em\relax ACM, 1994, pp. 226--235.

\bibitem{pollatos2008social}
G.~G. Pollatos, O.~A. Telelis, and V.~Zissimopoulos, ``On the social cost of
  distributed selfish content replication,'' in \emph{NETWORKING 2008 Ad Hoc
  and Sensor Networks, Wireless Networks, Next Generation Internet}.\hskip 1em
  plus 0.5em minus 0.4em\relax Springer, 2008, pp. 195--206.

\bibitem{gopalakrishnan2012cache}
R.~Gopalakrishnan, D.~Kanoulas, N.~N. Karuturi, C.~P. Rangan, R.~Rajaraman, and
  R.~Sundaram, ``Cache me if you can: capacitated selfish replication games,''
  in \emph{LATIN 2012: Theoretical Informatics}.\hskip 1em plus 0.5em minus
  0.4em\relax Springer, 2012, pp. 420--432.

\bibitem{goemans2006market}
M.~X. Goemans, L.~Li, V.~S. Mirrokni, and M.~Thottan, ``Market sharing games
  applied to content distribution in ad hoc networks,'' \emph{Selected Areas in
  Communications, IEEE Journal on}, vol.~24, no.~5, pp. 1020--1033, 2006.

\bibitem{chun2004selfish}
B.-G. Chun, K.~Chaudhuri, H.~Wee, M.~Barreno, C.~H. Papadimitriou, and
  J.~Kubiatowicz, ``Selfish caching in distributed systems: a game-theoretic
  analysis,'' in \emph{Proceedings of the twenty-third annual ACM symposium on
  Principles of distributed computing}.\hskip 1em plus 0.5em minus 0.4em\relax
  ACM, 2004, pp. 21--30.

\bibitem{goyal2000learning}
S.~Goyal and F.~Vega-Redondo, ``Learning, network formation and coordination,''
  Tinbergen Institute, Tech. Rep., 2000.

\bibitem{koutsoupias1999worst}
E.~Koutsoupias and C.~Papadimitriou, ``Worst-case equilibria,'' in \emph{STACS
  99}.\hskip 1em plus 0.5em minus 0.4em\relax Springer, 1999, pp. 404--413.

\bibitem{vetta2002nash}
A.~Vetta, ``{N}ash equilibria in competitive societies, with applications to
  facility location, traffic routing and auctions,'' in \emph{Foundations of
  Computer Science, 2002. Proceedings. The 43rd Annual IEEE Symposium
  on}.\hskip 1em plus 0.5em minus 0.4em\relax IEEE, 2002, pp. 416--425.

\bibitem{laoutaris2006distributed}
N.~Laoutaris, O.~Telelis, V.~Zissimopoulos, and I.~Stavrakakis, ``Distributed
  selfish replication,'' \emph{Parallel and Distributed Systems, IEEE
  Transactions on}, vol.~17, no.~12, pp. 1401--1413, 2006.

\bibitem{etesami2014pure}
S.~R. Etesami and T.~Ba\c{s}ar, ``Pure {N}ash equilibrium in capacitated
  selfish replication ({CSR}) game,'' \emph{arXiv preprint arXiv:1404.3442},
  2014.

\bibitem{etesami2015approximation}
------, ``A quasi-polynomial approximation algorithm for the {CSR} game and a
  tight bound on its price of anarchy,'' \emph{Submitted to 54rd IEEE
  Conference on Decision and Control. CDC'15. arXiv preprint arXiv:1412.6546},
  2015.

\bibitem{etesami2014game}
------, ``Game-theoretic analysis of the {H}egselmann-{K}rause model for
  opinion dynamics in finite dimensions,'' \emph{arXiv preprint
  arXiv:1412.6546}, 2014.

\bibitem{fabrikant2003network}
A.~Fabrikant, A.~Luthra, E.~Maneva, C.~H. Papadimitriou, and S.~Shenker, ``On a
  network creation game,'' in \emph{Proceedings of the twenty-second annual
  symposium on Principles of distributed computing}.\hskip 1em plus 0.5em minus
  0.4em\relax ACM, 2003, pp. 347--351.

\end{thebibliography}
\end{document}